\documentclass[a4paper,11pt,reqno]{amsart}
\usepackage[a4paper]{geometry}
\usepackage{amssymb,amsmath}
\usepackage{mathtools}
\usepackage{enumerate}
\usepackage{bookmark}
\usepackage{hyperref}
\usepackage[initials,backrefs]{amsrefs}

\numberwithin{equation}{section}
\theoremstyle{plain}
\newtheorem{theorem}{Theorem}

\newtheorem{lemma}{Lemma}

\theoremstyle{definition}
\newtheorem{definition}{Definition}

\newtheorem*{assi*}{(I) Short-range interaction}
\newtheorem*{asspone*}{(P1) Log-H\"older continuity condition}
\newtheorem*{assptwo*}{(P2) Rosenblatt strongly mixing condition}

\newtheorem*{dskn*}{$\dskn$}
\newtheorem*{dsknn*}{$\dsknn$}
\theoremstyle{remark}

\newcommand{\prob}[1]{\DP\left\{#1\right\}}
\newcommand{\esm}[1]{\mathbb{E}\left[\,#1\,\right]}
\newcommand{\Bone}{\mathbf{1}}

\newcommand{\BC}{\mathbf{C}}

\newcommand{\BG}{\mathbf{G}}
\newcommand{\BH}{\mathbf{H}}

\newcommand{\BK}{\mathbf{K}}

\newcommand{\BP}{\mathbf{P}}
\newcommand{\BU}{\mathbf{U}}

\newcommand{\BV}{\mathbf{V}}

\newcommand{\BX}{\mathbf{X}}

\newcommand{\CE}{\mathcal{E}}

\newcommand{\DN}{\mathbb{N}}
\newcommand{\DP}{\mathbb{P}}
\newcommand{\DR}{\mathbb{R}}
\newcommand{\DZ}{\mathbb{Z}}
\newcommand{\BDelta}{\mathbf{\Delta}}
\newcommand{\BPsi}{\mathbf{\Psi}}
\newcommand{\Bx}{\mathbf{x}}
\newcommand{\By}{\mathbf{y}}

\newcommand{\Bu}{\mathbf{u}}

\newcommand{\FB}{\mathfrak{B}}

\DeclareMathOperator{\dist}{dist}
\DeclareMathOperator{\supp}{supp}

\newcommand{\ee}{\mathrm{e}}

\newcommand{\ess}{\mathrm{ess}}

\newcommand{\condI}{\mathbf{(I)}}
\newcommand{\condPone}{\mathbf{(P1)}}
\newcommand{\condPtwo}{\mathbf{(P2)}}

\newcommand{\dskn}{\mathbf{(DS.}k,N\mathbf{)}}
\newcommand{\dsknn}{\mathbf{(DS.}k,n,N\mathbf{)}}

\begin{document}
\title[localization for continuous models with correlated potentials]{Localization for $N$-particle continuous models with strongly mixing correlated random potentials}

\author[T.~Ekanga]{Tr\'esor EKANGA$^{\ast}$}

\address{$^{\ast}$%
Institut de Math\'ematiques de Jussieu, 
Universit\'e Paris Diderot,
Batiment Sophie Germain, 
13 rue Albert Einstein,
75013 Paris,
France}
\email{tresor.ekanga@imj-prg.fr}
\date{\today}
\begin{abstract}
For the multi-particle Anderson model with correlated random potential in the continuum, we show under fairly general assumptions on the inter-particle interaction and the random external potential, the Anderson localization which consists of both the spectral, exponential localization and the strong dynamical localization. The localization results are proven near the lower spectral edge of the almost sure spectrum and the proofs require the uniform log-H\"older continuity assumption of the probability distribution functions of the random field in addition of the Rosenblatt's strongly mixing condition.
\end{abstract}
\maketitle

\section{Introduction, assumptions and the main result}

\subsection{Introduction}
We analyze multi-particle random Schr\"odinger operators in the continuous space of configurations. The work follows the paper \cite{E17} where the analysis was  done on the lattice. The main problem of the paper is that, we allow the values of the external random field to be correlated but strongly mixing. This results in a substantial modification of the scaling analysis in order to prove the localization results. 

Localization for correlated potentials was obtained by von Dreifus and Klein \cite{DK91} for single-particle models with Gaussian and completely analytic Gibbs fields. Later, Chulaevsky himself \cite{CS08} on the one hand and Boutet de Monvel \cites{BCSS10,BCS11} on the other hand  proved the Wegner estimates for correlated potentials and obtained in the sequel the Anderson localization. Also, Klopp \cite{Kl12}, analyzed the spectral statistic of the Andeson model in the continuum with weakly correlated random potentials. 

Let us recall that we assume two important assumptions on the correlated random variables in the Anderson model. First, the uniform log-H\"older continuity condition of the probability distribution functions of the random field.  This assumption is the one that makes the Wegner estimates of \cite{BCSS10} in a form suitable for the multi-scale analysis. Second, the Rosenblatt's strongly mixing condition which plays an important role in both a large deviation bound of the random stochastic process and the scale induction step of the multi-scale analysis.

As was said above, we will use the multi-scale analysis technique to prove our localization results, following the scheme developed for multi-particle models by Chulaevsky and Suhov in \cite{CS09} in the high disorder limit and adapted and improved in \cites{E11,E13} under the low energy regime. We recall that in the work \cite{E17} as well as in \cites{KN13,KN14},  the same results were proved in the i.i.d. case, so in this paper, all the proofs using independence must be revisited and re-written. 

Below, we describe the model and the assumptions. Our main result is Theorem \ref{thm:main.result} stated in Section 1.3. In Section 2, we prove a large deviation  bound for our multi-particle model with correlated random potential. Section 3 is devoted to the multi-particle multi-scale analysis at low energy. Finally, in section 4, we prove the localization results. Some parts of the rest of the text overlap with \cite{E17}.

\subsection{The model and the assumptions}

We fix the number of particles $N\geq 2$. We are concern with multi-particle random Schr\"odinger operators of the following form:
\[
\BH^{(N)}(\omega):=-\BDelta + \BU+\BV, 
\]
acting in $L^{2}((\DR^{d})^N)$. Sometimes, we will use the identification $(\DR^{d})^N\cong \DR^{Nd}$. Above, $\BDelta$ is the Laplacian on $\DR^{Nd}$, $\BU$ represents the inter-particle interaction which acts as multiplication operator in $L^{2}(\DR^{Nd})$. Additional information on $\BU$ is given in the assumptions. $\BV$ is the multi-particle random external potential also acting as multiplication operator on $L^{2}(\DR^{Nd})$. For $\Bx=(x_1,\ldots,x_N)\in(\DR^{d})^N$, $\BV(\Bx)=V(x_1)+\cdots+ V(x_N)$ and $\{V(x,\omega), x\in\DR^d\}$ is a random stochastic process relative to some probability space $(\Omega,\FB,\DP)$.

Observe that the non-interacting Hamiltonian $\BH^{(N)}_0(\omega)$ can be written as a tensor product:
\[
\BH^{(N)}_0(\omega):=-\BDelta +\BV=\sum_{k=1}^N \Bone^{\otimes(k-1)}_{L^{2}(\DR^d)}\otimes H^{(1)}(\omega)\otimes \Bone^{\otimes(N-k)}_{L^2(\DR^d)},
\]
where, $H^{(1)}(\omega)=-\Delta + V(x,\omega)$ acting on $L^2(\DR^d)$. We will also consider random Hamiltonian $\BH^{(n)}(\omega)$, $n=1,\ldots,N$ defined similarly. Denote by $|\cdot|$ the max-norm in $\DR^{nd}$.

\begin{assi*} 

Fix any $n=1,\ldots,N$. The potential of inter-particle interaction $\mathbf{U}$ is bounded and of the form
\[
\BU(\Bx)=\sum_{1\leq i<j\leq n}\Phi(|x_i-x_j|),\quad \Bx=(x_1,\ldots,x_n),
\]
where  $\Phi:\DN:\rightarrow\DR$ is a function such that

\begin{equation}\label{eq:finite.range.k}
\exists r_0\in\DN: \supp \Phi\subset[0,r_0].
\end{equation}
\end{assi*}

The random field $\{V(x,\omega); x\in\DZ^d\}$ is measurable with respect to some probability space $(\Omega,\FB,\DP)$. We define 
\[
F_{V,x}(t):=\prob{V(x,\omega)\leq t}\qquad \text{ and } F_{V,x}(t\big|\FB_{\neq x}):=\prob{V(x,\omega)\leq t\big| \FB_{\neq x}},
\]
the conditional probability distribution functions of $V$ where $\FB_{\neq x}$ represents the sigma-algebra generated by the random variables $\{V(y,\omega); y\neq x\}$

\begin{asspone*}
It is assumed that the conditional distribution functions $F_{V,x}$ are uniformly Log-H\"older continuous: for some $\kappa>0$ and any $\varepsilon >0$,
\[
\ess \sup_{x\in\DZ^d} \sup_{t\in\DR} \left(F_{V,x}(t+\varepsilon\big|\FB_{\neq x})-F_{V,x}(t\big|\FB_{\neq x})\right)\leq Const\cdot |\ln(\varepsilon)|^{\kappa}.
\]
\end{asspone*}

\begin{assptwo*} 
Let $L>0$ and positive constants $C_1>0$, $C_2>0$. For any pair of subsets $\Lambda', \Lambda'' \subset \DZ^{d}$ with $\dist(\Lambda',\Lambda'')\geq L$ and any events $\CE'\in\FB_{\Lambda'}$, $\CE''\in\FB_{\Lambda''}$,
\[
\left|\prob{\CE'\cap\CE''}-\prob{\CE'}\prob{\CE''}\right|\leq \ee^{-C_1 L}.
\]
Further, for any integer $\ell\geq 2$, and random variables $X_1(\omega),\ldots X_{\ell}(\omega)$, we have that
\[
\left|\esm{X_1\cdots X_{\ell}} - \esm{X_1}\cdots \esm{X_{\ell}}\right| \leq \ee^{-C_2 L^d},
\]
with $C_2> 3^d$.
\end{assptwo*}
Assumption $\condPtwo$ was used by Chulaevsky in \cite{C16} in the framework of his so-called Direct scaling of the multi-scale analysis under the high disorder regime. Above, 
$\FB_{\Lambda'}$ and $\FB_{\Lambda''}$ are the sigma-algebra generated  by the random variables $\{V(x,\omega); x\in\Lambda'\}$ and $\{V(x,\omega); x\in\Lambda''\}$ respectively.

\subsection{The main result}

\begin{theorem}\label{thm:main.result}
Assume that the hypotheses $\condI$, $\condPone$ and $\condPtwo$ hold true. Then

\begin{enumerate}
\item[A)] 
The lower spectral edge $E_0^{(N)}$ of $\BH^{(N)}(\omega)$, is almost surely non-random and there exist $E^*> E_0^{(N)}$ such that the spectrum of $\BH^{(N)}(\omega)$ in $[E_0^{(N)}, E^*]$ is pure point and each eigenfunction corresponding to eigenvalues in $[E_0^{(N)},E^*]$ is exponentially decaying at infinity in the max-norm.\\

\item[B)]
There exist $E^*>E^{(N)}_0$  such that for any bounded domain $\BK\subset \DR^{Nd}$, we have 
\begin{equation}\label{eq:low.energy.dynamical.loc}
\esm{\sup_{t>0}\| \BX^{s}\ee^{-it\BH^{(N)}(\omega)}\BP_I(\BH^{(N)}(\omega))\Bone_{\BK}\|_{L^2(\DR^{Nd})}}<\infty,
\end{equation}
where $(|\BX|\BPsi)(\Bx):=|\Bx|\BPsi(\Bx)$, $\BP_{I}(\BH^{(N)}(\omega))$ is the spectral projection of $\BH^{(N)}(\omega)$ onto the interval $I:=[E^{(N)}_0,E^*]$, and the supremum is taken over bounded measurable functions $f$. 
\end{enumerate}
\end{theorem}

\section{Geometry and large deviation estimates}
 For $\Bu=(u_1,\ldots,u_n)\in\DZ^{nd}$, we denote by $\BC^{(n)}_L(\Bu)$ the $n$-particle open cube, i.e,
\[
\BC^{(n)}_L(\Bu)=\left\{\Bx\in\DR^{nd}:|\Bx-\Bu|< L\right\},
\]
and given $\{L_i: i=1,\ldots,n\}$, we define the rectangle
\begin{equation}          \label{eq:cube}
\BC^{(n)}(\Bu)=\prod_{i=1}^n C^{(1)}_{L_i}(u_i),
\end{equation}
where $C^{(1)}_{L_i}(u_i)$ are cubes of side length $2L_i$ center at points $u_i\in\DZ^d$. We also define 
\[
\BC^{(n,int)}_L(\Bu):=\BC^{(n)}_{L/3}(\Bu), \quad \BC^{(n,out)}_L(\Bu):=\BC^{(n)}_L(\Bu)\setminus\BC^{(n)}_{L-2}(\Bu), \quad \Bu\in\DZ^{nd}
\]
and introduce  the characteristic functions:
\[
\Bone^{(n,int)}_{\Bx}:=\Bone_{\BC^{(n,int)}_L(\Bx)}, \qquad \Bone^{(n,out)}_{\Bx}:= \Bone_{\BC^{(n,out)}_L(\Bx)}.
\]
The volume of the cube $\BC^{(n)}_L(\Bu)$ is $|\BC_L^{(n)}(\Bu)| :=(2L)^{nd}$.
We denote the restriction of the Hamiltonian $\BH^{(n)}$ to  $\BC^{(n)}(\Bu)$ by
\begin{align*}
&\BH_{\BC^{(n)}(\Bu)}^{(n)}=\BH^{(n)}\big\vert_{\BC^{(n)}(\Bu)}\\
&\text{with Dirichlet boundary conditions}
\end{align*}

We denote the spectrum of $\BH_{\BC^{(n)}(\Bu)}^{(n)}$  by
$\sigma\bigl(\BH_{\BC^{(n)}(\Bu)}^{(n)}\bigr)$ and its resolvent by
\begin{equation}\label{eq:def.resolvent}
\BG^{(n)}_{\BC^{(n)}(\Bu)}(E):=\Bigl(\BH_{\BC^{(n)}(\Bu)}^{(n)}-E\Bigr)^{-1},\quad E\in\DR\setminus\sigma\Bigl(\BH_{\BC^{(n)}(\Bu)}^{(n)}\Bigr).
\end{equation}

\begin{definition}
Let $m>0$ and $E\in\DR$ be given.  A cube $\BC_L^{(n)}(\Bu)\subset\DR^{nd}$, $1\leq n\leq N$ will be  called $(E,m)$-\emph{nonsingular} ($(E,m)$-NS) if $E\notin\sigma(\BH^{(n)}_{\BC^{(n)}_{L}(\Bu)})$ and
\begin{equation}\label{eq:singular} 
\|\Bone^{(n,out)}_{\Bx}\BG^{(n)}_{\BC^{(n)}_L(\Bx)}(E)\Bone^{(n,int)}_{\Bx}\|\leq\ee^{-\gamma(m,L,n)L},
\end{equation}
where
\begin{equation}\label{eq:gamma}
\gamma(m,L,n)=m(1+L^{-1/8})^{N-n+1}.                     
\end{equation}
Otherwise it will be called $(E,m)$-\emph{singular} ($(E,m)$-S).
\end{definition}

We prove in this subsection an analog of the large deviation estimate of \cite{St01} in the case of correlated potentials under the assumption $\condPtwo$. 

\begin{lemma}\label{lem:DEV} 
Let $L>0$ and set $s_0:=\min_{x\in C^{(1)}_L(0)}\{-\frac{1}{2}\ln\esm{\exp(-V(x,\omega))}\}>0$.
Under assumption $\condPtwo$, we have that:
\[
\prob{\frac{1}{|C^{(1)}_L(0)|} \sum_{x\in C^{(1)}_L(0)} \sum_{x\in C^{(1)}_L(0)} V(x,\omega)\leq s_0}\leq \exp(-\gamma_0|C^{(1)}_L(0)|),
\]
for some $\gamma_0>0$.
\end{lemma}

\begin{proof}
Since each quantity $V(x,\omega)$ is non-negative, we have that $\esm{\exp(-V(x,\omega))}<1$. Thus, setting $\gamma_x=-\ln(\esm{\exp(-V(x,\omega))})>0$ and using the Rosenblatt's strongly mixing condition $\condPtwo$, we have that:

\begin{gather*}
\prob{\frac{1}{| C^{(1)}_L(0)|}\sum_{x\in C^{(1)}_L(0)} V(x,\omega)\leq s_0 }\\
=\prob{\sum_{x\in C^{(1)}_L(0)} V(x,\omega)\leq s_0 |C^{(1)}_L(0)|}\\
=\prob{\exp(s_0 |C^{(1)}_L(0)|- \sum_{x\in C^{(1)}_L(0)} V(x,\omega))\geq 1}\\
\leq \exp(s_0|C^{(1)}_L(0)|)\cdot \left(\prod_{x\in C^{(1)}_L(0)}\esm{\exp(-V(x,\omega))}+\ee^{-C_2L^d}\right)\\
\leq \exp\left(s_0|C^{(1)}_L(0)| - \sum_{x\in C^{(1)}_L(0)} \gamma_x -C_2L^d\right)\\
\leq \exp\left( \sum_{x\in C^{(1)}_L(0)} s_0 - \gamma_x-C_2L^d\right)\\
\leq \exp\left(-s_0|C^{(1)}_L(0)|\right)\times\ee^{-C_2L^d},\\
\leq \exp\left(-\gamma_0|C^{(1)}_L(0)|\right),
\end{gather*}
for some $\gamma_0>0$ and $L>0$ large enough. Indeed, $s_0\leq \frac{1}{2}\gamma_x$ for all $x\in C^{(1)}_L(0)$.

\end{proof}

Now, below, we give an important result on the first eigenvalue for the single-particle Hamiltonian:

\begin{lemma}\label{lem:low.energy.gap.prob}
Assume that assumption $\condPtwo$ holds true. There exist $b>0$ and $\gamma>0$ such that 
\[
\prob{E^{(1)}_0(\omega)\leq b L^{-2}}\leq \ee^{-\gamma L^d},
\]
where $E^{(1)}_0(\omega)$ denotes  the infimum of $\sigma(H^{(1)}_{C^{(1)}_{L}(0)}(\omega))$.
\end{lemma}

\begin{proof}
See the proof of Theorem 2.1.3 in \cite{St01} which is based on the empirical average bound given in Lemma \ref{lem:low.energy.gap.prob}.
\end{proof}

Now, it is straightforward to show that the same result holds true for the multi-particle random Hamiltonian.

\begin{theorem}\label{thm:np.bottom}

Under hypothesis $\condPtwo$, for any $p>0$, there exists $L^*_1>0$ such that 
\[
\prob{E_0^{(n)}(\omega)\leq L^{-1/2}} \leq L^{-2p4^{N-n}},
\]
for all $L\geq L^*_1$.
\end{theorem}

\begin{proof}
We denote by $\BH^{(n)}_0(\omega)$ the multi-particle random Hamiltonian without interaction. Observe that, since the interaction potential $\BU$ is non-negative, we have 
\[
E_0(\BH^{(n)}_{\BC^{(n)}_L(\Bu)}(\omega)\geq  E_0^{(n)}(\omega),
\]
where $E_0^{(n)}(\omega)=\lambda_1^{(1)}(\omega)+\cdots+\lambda_n^{(1)}$ and the $\lambda_i^{(1)}(\omega)$ are the eigenvalues of the single-particle random Hamiltonians $H^{(1)}_{\BC^{(1)}_L(u_i)}(\omega)$, $i=1,\ldots,n$. So, if $E^{(n)}_0(\omega)\leq L^{-1/2}$, then for example $\lambda_1^{(1)}(\omega)\leq L^{-1/2}$ and this implies the required probability bound of the assertion of Theorem \ref{thm:np.initial.MSA}.
\end{proof}

\section{The multi-particle multi-scale analysis}

It is convenient here to recall the Combes-Thomas estimate.

\begin{theorem}\label{thm:CT}
Let $H=-\BDelta + W$ be a Schr\"odinger operator on $L^2(\DR^{D})$, $E\in\DR$ and $E_0=\inf \sigma(H)$. Set $\eta=\dist(E,\sigma(H))$. If $E<E_0$,, then for any $0<\gamma<1$, we have that:
\[
\left\|\Bone_{x}(H-E)^{-1}\Bone_{y}\right\|\leq \frac{1}{(1-\gamma^2)\eta}\ee^{\gamma\sqrt{\eta d}}\ee^{-\gamma \sqrt{\eta}|x-y|},
\]
for all $x,y\in\DR^{D}$. 
\end{theorem}  

\begin{proof}
See the proof of Theorem $1$ in \cite{GK02}.
\end{proof}

Recall that the parameter $m>0$ is given by $m=\frac{2^{-N}\gamma L^{-1/4}}{3\sqrt{2}}$.

\begin{theorem}\label{thm:np.initial.MSA}

Assume that the hypotheses $\condI$, $\condPone$ and $\condPtwo$ hold true. Then, there exists $E^*>0$ such that 
\[
\prob{\text{$\exists E\in (-\infty: E^*]:$ $\BC^{(n)}_L(\Bu)$ is $(E,m)$-S}}\leq L^{-2p4^{N-n}},
\]
for $L>0$ large enough.
\end{theorem}

\begin{proof}
Set $E^*:=\frac{1}{2} L_0^{-1/2}$. If the first eigenvalue $E_0^{(n)}(\omega)$ satisfies $E_0^{(n)}(\omega)>L^{-1/2}$, then for all energy $E\leq E^*$, we have:

\begin{align*}
\dist(E,\sigma(\BH^{(n)}_{\BC^{(n)}_L(\Bu)}))&=E^{(n)}_0(\omega)\\
&> L^{-1/2}-\frac{1}{2} L^{-1/2}\\
&>\frac{1}{2} L^{-1/2}
\end{align*}
Thus by the Combes Thomas estimate Theorem \ref{thm:CT},
\begin{align*}
\|\Bone_{\Bx}\BG^{(n)}_{\BC^{(n)}_L(\Bu)}(E)\Bone_{\By}\|&\leq 2L^{-1/2}\ee^{\gamma \sqrt{d}\sqrt{\eta}}\ee^{-\gamma \sqrt{\eta} |\Bx-\By|}\\
&\leq 2 L^{1/2}\ee^{-\frac{\gamma L^{-1/4}}{\sqrt{2}}(\frac{L}{3}-\sqrt{d})}
\end{align*}
 
Thus for $L>0$ large enough depending on the dimension $d$, we get
\begin{gather*}
\|\Bone_{\BC^{(n,out)}_L(\Bu)}\BG^{(n)}_{\BC^{(n)}_L(\Bu)}(E)\Bone_{\BC^{(n,int)}_L(\Bu)}\|\\
\leq \sum\limits_{{\substack{\Bx\in\BC^{(n,out)}_L(\Bu)\cap_\DZ^{nd}\\ \By\in\BC^{(n,int)}_L(\Bu)\cap\DZ^{nd}}}} 2L^{-1/2}\ee^{-\frac{\gamma L^{-1/4}}{\sqrt{2}}(\frac{L}{3}-\sqrt{d})}\\
\leq (2L)^{2nd} 2 L^{1/2}\ee^{-2^N m L}\\
\end{gather*}

Now since $\gamma(m,L,n)=m(1L^{-1/8})^{N-n}<2^N m$, for $L>0$, large enough, we have that
\[
\|\Bone_{\BC^{(n,out)}_L(\Bu)}\BG^{(n)}_{\BC^{(n)}_L(\Bu)}(E)\Bone_{\BC^{(n,int)}_L(\Bu)}\|\leq \ee^{-\gamma(m,L,n)L}.
\]
The above analysis, then implies that

\begin{gather*}
\prob{\text{$\exists E\leq E^*$: $\BC^{(n)}_L(\Bu)$ is $(E,m)$-S}}\\
\leq \prob{E^{(n)}_0(\omega)\leq L^{-1/2}}\leq L^{-2p 4^{N-n}},
\end{gather*}
Yielding the required result.
\end{proof}

Now the rest of the multi-particle multi-scale analysis can be done exactly in the same way as in our earlier work \cite{E17} in the case of i.i.d. random potential.    

\section{Proof of the main result}
Using the multi-scale analysis bounds from the above Section, the localization result can be proved in the same way as in the paper \cite{E17} for i.i.d. random  external potentials. Also see \cite{BCS11}.

\end{document}